\documentclass[twoside]{article}
\usepackage{amsfonts,amssymb,amsbsy,textcomp,marvosym,picins,amsmath,caption,threeparttable,amsthm,subfigure,float,lastpage,lscape}
\usepackage{eurosym,mathrsfs,fancyhdr,CJK,multicol,graphics,indentfirst,color,bm,upgreek,booktabs,graphicx,multirow,warpcol}
\usepackage{epstopdf}
\usepackage{mathtools}
\def\footnoterule{\kern 1mm \hrule width 10cm \kern 2mm}

\allowdisplaybreaks
\catcode`@=11
\def\title#1{\vspace{3mm}\begin{flushleft}\vglue-.1cm\Large\bf\boldmath\protect\baselineskip=18pt plus.2pt minus.1pt #1
\end{flushleft}\vspace{1mm} }
\def\author#1{\begin{flushleft}\normalsize #1\end{flushleft}\vspace*{-4pt} \vspace{3mm}}

\def\jz#1#2{{$^{\footnotesize\textcircled{\tiny #1}}$\let\thefootnote\relax\footnotetext{\!\!$^{\footnotesize\textcircled{\tiny #1}}$#2}}}
\catcode`@=11
\def\section{\@startsection{section}{1}{\z@}%
 {-3ex \@plus -.3ex \@minus -.2ex}%
 {2.2ex \@plus.2ex}%
{\normalfont\normalsize\protect\baselineskip=14.5pt plus.2pt minus.2pt\bfseries}}
\def\subsection{\@startsection{subsection}{2}{\z@}%
 {-3ex\@plus -.2ex \@minus -.2ex}%
 {2ex \@plus.2ex}%
{\normalfont\normalsize\protect\baselineskip=12.5pt plus.2pt minus.2pt\bfseries}}
\def\subsubsection{\@startsection{subsubsection}{3}{\z@}%
 {-2.2ex\@plus -.21ex \@minus -.2ex}%
 {1.4ex \@plus.2ex}
{\normalfont\normalsize\protect\baselineskip=12pt plus.2pt minus.2pt\sl}}

\newtheorem{theorem}{Theorem}

\DeclarePairedDelimiter\ceil{\lceil}{\rceil}

\begin{document}
\begin{CJK*}{GBK}{song}
\thispagestyle{empty}
\title{Freeze-Tag is NP-Hard in 3D with $L_1$ distance}

\author{Lucas de Oliveira Silva}

\vspace*{4mm}
\noindent {\small\bf Abstract} \quad  {\small Arkin et al. \cite{Arkin06} in 2002 introduced a scheduling-like problem called Freeze-Tag Problem (FTP) motivated by robot swarm activation. The input consists of the locations of $n$ mobile punctual robots in some metric space or graph. Only one begins ``active", while the others are initially ``frozen". All active robots can move at unit speed and, upon reaching a frozen one's location, activates it. The goal is to activate all the robots in the minimum amount of time, the so-called makespan. Until 2017 the hardness of this problem in metric spaces was still open, but then Yu et al. \cite{Yu17} proved it to be NP-Hard in the Euclidian plane, and in the same year, Demaine and Roudoy \cite{Erik17} demonstrated that the FTP is also hard in 3D with any $L_p$ distance (with $p>1$). However, we still don't know whether Demaine's and Roudoy's result could be translated to the plane. This paper fills the $p=1$ gap by showing that the FTP is NP-Hard in 3D with $L_1$ distance.}

\vspace*{3mm}
\noindent{\small\bf Keywords} \quad {\small Complexity, NP-Hardness, Scheduling, Computational Geometry, Swarm Robotics}
\vspace*{4mm}
\end{CJK*}

\baselineskip=18pt plus.2pt minus.2pt
\parskip=0pt plus.2pt minus0.2pt
\begin{multicols}{2}

\section{Introduction}
The Freeze-Tag Problem was introduced by Arkin et al. in 2002 \cite{Arkin06} to model the automatic awakening of a robot swarm by manually turning on just one of the individuals. The input consists of a list of $n$ mobile punctual robots' locations in some metric space or graph. Only one of them begins ``active" (``on", ``unfrozen" or ``awake") and is called the source, while the others are initially ``frozen" (``off", ``frozen" or ``asleep"). Frozen robots become active when an active one reaches its location. After activated, a robot can move at unit speed and help unfreeze the remaining frozen ones. In the decision version, we are also given a time limit, and the goal is to decide whether all the robots can be activated within this limit.

Depending on the choice of metric (defining ``unit speed"), domain, or some other restriction, the problem takes different forms. For example, Arkin et al. \cite{Arkin06} proved that the FTP is NP-Hard on Weighted Star Graphs, but is in P when restricted to the unweighted ones. Moreover, by Bender et al. \cite{Arkin03}, we know that it is hard on Unweighted General Graphs. The first ones to tackle the NP-Hardness in metric spaces were Yu et al. at FWCG 2017 \cite{Yu17}, who proved the problem NP-Hard in the Euclidian plane, and later Demaine with Rudoy, that showed that the problem is also hard in 3D with any $L_p$ distance (with $p>1$). Our result complements the work of Demaine and Rudoy by proving that the FTP is also NP-Hard in 3D with $L_1$ distance.

\section{Reduction}
To prove the hardness claim, we reduce from $Monotone$-$3SAT$, shown to be NP-Complete in \cite{Go78}. An instance of this problem consists of a set of boolean variables $\{x_1, \dots, x_n\}$ and clauses $\{c_1, \dots, c_m\}$ in which three literals form each clause, and each literal is a, possibly negated, copy of one of the variables. Each clause is either \emph{positive} when all its literals are not negated or \emph{negative} when all are negated. The question is whether a satisfying boolean assignment to the variables exists, that is, an assignment such that every clause contains at least one true literal. Throughout this paper, Latin uppercase letters will denote points and the lowercase robots.

\begin{theorem}
Freeze-Tag is NP-Complete in 3D with $L_1$ distance.
\end{theorem}

\begin{proof}
As shown in \cite{Yu16}, the problem is in NP.

Given an instance $X=\{x_1, \dots, x_n\}$, $C=\{c_1, \dots, c_m\}$ (with $n>1$ and $m>2$) of the $Monotone$-$3SAT$ we can assume, without loss of generality, that $n$ is odd and the number of positive and negative clauses are the same. If not, we can add three dummy variables and a single positive clause containing them and then duplicate some clauses until the number of each type is equal. So let $C = \{c_1, \dots, c_{m/2}\} \cup \{\overline{c}_1, \dots, \overline{c}_{m/2}\}$ where the overline represents a negative clause. We place robots as follows:

Let $\alpha=\frac{2}{n-1}$, $\epsilon<\alpha/4$, $\beta=\frac{2}{m-2}$ and $L=6+\epsilon$. At the origin we have $n$ robots, $r_1, r_2, \dots, r_n$, with $r_1$ being the source one. For each variable $x_{i+1}$, there are two groups of robots, a \emph{true} one $t_{i+1}$ at $T_{i+1} = (2-i\alpha; i\alpha; \epsilon)$, and a false one $f_{i+1}$ at $F_{i+1} = (2-i\alpha; i\alpha; -\epsilon)$. The first corresponds to the occurrences of $x_{i+1}$ as a not negated literal, and the second corresponds to the negated ones. So each group contains a robot for each occurrence of $x_{i+1}$ in a clause being not negated or negated, respectively. The awakening order of these groups will represent the boolean value of $x_{i+1}$, that is, $x_{i+1} = true$ iff $t_{i+1}$ is awakened before $f_{i+1}$. Additionally, we add another $n$ robots, $s_1, s_2, \dots, s_n$, each $s_{i+1}$ being located at $S_{i+1} = (2-i\alpha+(L-2-4\epsilon)/2; i\alpha+(L-2-4\epsilon)/2; 0)$.

Then for each $k=0, \dots, m/2-1$ and each variable $x_{i+1}$ that does appear (not) negated in clause ($c_{k+1}$) $\overline{c}_{k+1}$, we create four robots:

$\mathbf{c}^{i+1}_{k+1}$ at $\mathbf{C}^{i+1}_{k+1} = T_{i+1} + (k\beta; k\beta; 2-2k\beta)$,

$\mathbf{d}^{i+1}_{k+1}$ at $\mathbf{D}^{i+1}_{k+1} = \mathbf{C}^{i+1}_{k+1} + P$,

$\overline{\mathbf{c}}^{i+1}_{k+1}$ at $\overline{\mathbf{C}}^{i+1}_{k+1} = F_{i+1} + (k\beta; k\beta; -2+2k\beta)$ and

$\overline{\mathbf{d}}^{i+1}_{k+1}$ at $\overline{\mathbf{D}}^{i+1}_{k+1} = \overline{\mathbf{C}}^{i+1}_{k+1} - P$,

where $P = (0; 0; L - 4 - 3\epsilon)$.

Finally, for each $k=0, \dots, m/2-1$ we add two robots:

$\mathbf{c}_{k+1}$ at $\mathbf{C}_{k+1} = C^{\ceil{n/2}}_{k+1} + Q$ and

$\overline{\mathbf{c}}_{k+1}$ at $\overline{\mathbf{C}}_{k+1} = \overline{C}^{\ceil{n/2}}_{k+1} + Q$, where $Q = (1; 1; 0)$.

We will prove that there is a satisfying boolean assignment to the given $Monotone$-$3SAT$ instance iff this FTP instance has an awakening schedule, with $r_1$ as the source, of optimal makespan $L$. \break

$(\implies)$ Suppose that the given $Monotone$-$3SAT$ has a satisfying boolean assignment. Without loss of generality, suppose that $T=\{x_1, \dots, x_l\}$ are the variables assigned to $true$ and $F=X\backslash T = \{x_{l+1}, \dots, x_n\}$ the variables assigned to $false$. Consider this scheduling:

At time instant zero, all $r_i$'s become awake. Then for each $r_i \in T$, $r_i$ goes to group $t_i$ at point $T_i$, and for each $r_j \in F$, $r_j$ goes to group $f_j$ at point $F_j$.

These robots reach their target at $2+\epsilon$. Upon arrival, a number of robots equal to the number of occurrences of the corresponding target literal become awake.
Afterward, the $r_i$'s go to their opposite still frozen group, and the other robots inside that first group split paths. Each one at a $T_i$ point goes to a different $\mathbf{c}_k^i$ robot that $x_i$ appears in the corresponding $c_k$ clause, and those at a $F_i$ point analogously go to a different $\overline{\mathbf{c}}^i_k$.

The $r_i$'s will reach their new destination group at $2+3\epsilon$ when another similar splitting occurs, now seeking to awaken those missing $\mathbf{c}_k^i$'s and $\overline{\mathbf{c}}^i_k$'s. They will be unfrozen at $4+3\epsilon$ and go to the matching $\mathbf{D}^i_k$ ($\overline{\mathbf{D}}^j_l$) point above (below) arriving in $L$. After this, each $r_i$ heads directly to a $s_i$ robot, getting there at time $L$.

In contrast, by $4+\epsilon$, the first batch of $\mathbf{c}^i_k$'s and $\overline{\mathbf{c}}^j_l$'s had already become awake and went directly to the corresponding $\mathbf{c}_k$ ($\overline{\mathbf{c}}_k$). Those who awakened them are now going to the matching $\mathbf{D}^i_k$ ($\overline{\mathbf{D}}^j_l$) point above (below), getting there at $L-2\epsilon$ when they'll finish.

At $4+\epsilon$, at least one robot will be heading each $\mathbf{c}_k$ ($\overline{\mathbf{c}}_k$) since the boolean assignment is a satisfying one. Furthermore, all these remaining robots are unfrozen by $L$.

$(\impliedby)$ Suppose the constructed FTP instance has an awakening schedule, with $r_1$ as the source, of optimal makespan $L$. The $n$ $s_i$'s are each at a distance $L-4\epsilon$ from the origin, and the distance between them is at least $\alpha>4\epsilon$. So, without loss of generality, we can assume that for each $i$, $r_i$ awakens $s_i$. While going to $s_i$, each $r_i$ has $4\epsilon$ time to spare. So there is an optimal solution where for each $i$, $r_i$ awakens the groups $t_i$ and $f_i$ before going to $s_i$. Consider that we are in this case.

Therefore we can construct a boolean assignment to the given $Monotone$-$3SAT$ instance as follows: $x_i = true$ iff $t_i$ is awakened before $f_i$. \break

Without loss of generality, after being awakened, each robot at a $t_i$ ($f_i$) group must go directly to a different $\mathbf{d}_k^i$ ($\overline{\mathbf{d}}^i_k$) robot that $x_i$ appears in the corresponding $c_k$ ($\overline{c}_k$) clause as there are as many of these as themselves. Because we are dealing with $L_1$ distance, they can also unfreeze the corresponding $\mathbf{c}_k^i$ or $\overline{\mathbf{c}}^i_k$ in their path. Also we can assume that they do so, as this is the earliest that they could've become unfroze anyway.

Every $\mathbf{c}_k^i$ ($\overline{\mathbf{c}}^i_k$) is at distance $2$ from its corresponding $\mathbf{c}_k$ ($\overline{\mathbf{c}}_k$). The literals that are set to \emph{true} awakens the robots corresponding to the causes that they appear on at time $4+\epsilon$, but the other $\mathbf{c}_k^i$ and  $\overline{\mathbf{c}}^i_k$ robots are reached only at $4+3\epsilon$. So only the first batch can reach its corresponding $\mathbf{c}_k$ ($\overline{\mathbf{c}}_k$) in time. We've assumed that we are dealing with optimal scheduling of makespan $L$, so each $\mathbf{c}_k$ ($\overline{\mathbf{c}}_k$) is awakened in the end by a literal whose boolean value is \emph{true}.

Therefore the boolean assignment makes each clause contain at least one \emph{true} literal, making it a satisfying boolean assignment. \break

By these previous implications, the FTP in 3D with $L_1$ distance is NP-Hard, so by the first claim, its NP-Complete.
\end{proof}

\section{Conclusion}
The obvious open problem is whether the result of this paper could also be applied to the 2D case. Considering that the Euclidian counterpart is NP-Hard, the FTP in 2D with $ L_1 $ distance may also be.

\small
\bibliographystyle{alpha}
\bibliography{bib}

\end{multicols}
\end{document}